\documentclass[12pt]{article}
\usepackage{amsmath, amsthm, amssymb, amsfonts, enumerate,bm,url}
\usepackage[figuresright]{rotating}
\usepackage{natbib}
\usepackage{tikz}
\usetikzlibrary{shapes,arrows}

\begin{document}
\newcommand{\abs}[1]{\left\vert#1\right\vert}
\newcommand{\set}[1]{\left\{#1\right\}}
\newcommand{\eps}{\varepsilon}
\newcommand{\To}{\rightarrow}
\newcommand{\inv}{^{-1}}
\newcommand{\ihat}{\hat{\imath}}
\newcommand{\var}{\mbox{Var}}
\newcommand{\sd}{\mbox{SD}}
\newcommand{\cov}{\mbox{Cov}}
\newcommand{\f}{\frac}
\newcommand{\fI}[1]{\frac{1}{#1}}
\newcommand{\what}[1]{\widehat{#1}}
\newcommand{\hhat}[1]{\what{\what{#1}}}
\newcommand{\wtilde}[1]{\widetilde{#1}}
\newcommand{\bdot}{\bm{\cdot}}
\newcommand{\Th}{\theta}
\newcommand{\qmq}[1]{\quad\mbox{#1}\quad}
\newcommand{\qm}[1]{\quad\mbox{#1}}
\newcommand{\mq}[1]{\mbox{#1}\quad}
\newcommand{\tr}{\mbox{tr}}
\newcommand{\logit}{\mbox{logit}}
\newcommand{\noi}{\noindent}
\newcommand{\bni}{\bigskip\noindent}
\newcommand{\bul}{$\bullet$ }
\newcommand{\bias}{\mbox{bias}}
\newcommand{\conv}{\mbox{conv}}
\newcommand{\spn}{\mbox{span}}
\newcommand{\colspace}{\mbox{colspace}}
\newcommand{\mA}{\mathcal{A}}
\newcommand{\mD}{\mathcal{D}}
\newcommand{\mF}{\mathcal{F}}
\newcommand{\mH}{\mathcal{H}}
\newcommand{\mI}{\mathcal{I}}
\newcommand{\mN}{\mathcal{N}}
\newcommand{\mR}{\mathcal{R}}
\newcommand{\mT}{\mathcal{T}}
\newcommand{\mX}{\mathcal{X}}
\newcommand{\bbR}{\mathbb{R}}
\newcommand{\vphi}{\varphi}

% footnote symbols
\renewcommand*{\thefootnote}{\fnsymbol{footnote}}

\newtheorem{theorem}{Theorem}[section]
\newtheorem{corollary}{Corollary}[section]
\newtheorem{conjecture}{Conjecture}[section]
\newtheorem{proposition}{Proposition}[section]
\newtheorem{lemma}{Lemma}[section]
\newtheorem{definition}{Definition}[section]
\newtheorem{example}{Example}[section]
\newtheorem{remark}{Remark}[section]

%roman numbers
\makeatletter
\newcommand{\romannum}[1]{\romannumeral #1}
\newcommand{\Romannum}[1]{\expandafter\@slowromancap\romannumeral #1@}
\makeatother

%draw flowchart
\tikzstyle{decision} = [diamond, draw, text width=5em, text badly centered, node distance=3cm, inner sep=0pt]
\tikzstyle{block} = [rectangle, draw, text width=10em, text centered, rounded corners, minimum height=2em]
\tikzstyle{line} = [draw, -latex']

\title{{\bf\Large A Rejection Principle for Sequential Tests of Multiple Hypotheses Controlling Familywise Error Rates}}

\author{\textsc{Jay Bartroff\footnote{Department of Mathematics, University of Southern California, Los Angeles, California, USA. Email: \texttt{bartroff@usc.edu}.} and Jinlin Song\footnote{Analysis Group, Inc., Boston, Massachusetts.}}}

\date{}
\maketitle

\abstract{We present a unifying approach to multiple testing  procedures for sequential (or streaming) data by giving sufficient conditions for a sequential multiple testing procedure to control the familywise error rate (FWER). Together we call these conditions a ``rejection 
principle for sequential tests,'' which we then apply to some existing sequential multiple testing procedures to give simplified understanding of their FWER control. Next the principle is applied to derive two new sequential multiple testing procedures with
provable FWER control, one for testing hypotheses in order and another for
closed testing. Examples of these new procedures are given by
applying them to a chromosome aberration data set and to finding the
maximum safe dose of a treatment.} 

\section{Introduction and background}\label{sec:intro}

The need for multiple-comparison-type corrections due to testing more than
one null hypothesis occurs in nearly all areas of scientific inquiry in
which statistical hypothesis testing is employed. In a number of these
areas, the data is inherently sequential\footnote[3]{The term
``sequential'' works overtime in the statistics literature, being employed
to describe both the sequential analysis of data
\citep[e.g.,][]{Siegmund85} as well as the step-wise analysis of
fixed sample size test statistics that occurs in many
multiple testing procedures \citep[e.g.,][]{Goeman10}.  For clarity,
herein we reserve the term ``sequential'' to describe the former and we
use ``fixed sample size'' when we want to emphasize the latter.}, or
``streaming,'' such as in multiple endpoint (or multi-arm) clinical trials
\citep[][Chapter~15]{Jennison00}, multi-channel changepoint detection
\citep{Tartakovsky03} and its applications to biosurveillance
\citep{Mei10}, genetics and genomics \citep[e.g.,][]{Salzman11}, and
acceptance sampling with multiple criteria \citep{Baillie87}. Adopting the
familywise error rate (FWER) metric, this paper takes a unifying approach
to sequential multiple testing procedures that control the FWER.  
Specifically, we give sufficient conditions for a sequential multiple
testing procedure to control the FWER, which turn out to be much simpler
and easier to verify in many cases than comprehensive analysis of the
procedure. We call these two sufficient conditions, given in Theorem~\ref{thm:rej.princ}, a \emph{rejection principle
for sequential tests}, following and extending the seminal work of
\citet{Goeman10} who accomplished this for fixed sample size procedures and in
turn extended and unified the work of \citet{Romano05}, \citet{Hommel07},
and \citet{Marcus76}.  Two overlapping aspects of the problem that we must
deal with in the sequential setting
 that were absent from the fixed sample size setting are how to allow for acceptances of hypotheses as well as rejections, and the interplay of sequential sampling with the accept/reject decisions. In the fixed sample size 
setting, rejecting a hypothesis is equivalent to not accepting it, however
in the sequential setting these are not necessarily equivalent because there is the
third possibility of performing additional sampling. These aspects are dealt with by expanding the notion of a procedure's rejection function, introduced in Section~\ref{sec:theorem}, to incorporate not just the already rejected hypotheses as in Goeman and Solari's \citeyearpar{Goeman10} fixed sample size setting, but also the already accepted hypotheses as well as the current sample size of those data streams that are still being sampled. 

The rest of the paper is organized as follows.  After briefly reviewing the relevant background in the next paragraph, our rejection principle is introduced in Section~\ref{sec:theorem} and its sufficiency for FWER control is established in Theorem~\ref{thm:rej.princ}. In Section~\ref{sec:apps} we apply our rejection principle to derive sequential multiple testing procedures, first deriving two general procedures that do not assume a special structure among the hypotheses that control the FWER and both type I and II FWERs (defined below), respectively.  Then our rejection principle is applied to derive sequential procedures for two settings wherein special structure of the hypothesis is known: testing hypotheses in order \citep{Rosenbaum08}  and closed testing \citep{Marcus76}. In Section~\ref{sec:ex} we give examples of these derived procedures, first applying the sequential procedure for testing hypotheses in order to real data from a study \citep{Masjedi00} of  chromosome aberration effects of an anti-tuberculosis drug, and then applying the closed testing procedure to finding the maximum safe dose of a treatment, wherein the sequential procedure is evaluated in a simulation study. Section~\ref{sec:disc} provides a summary and discussion.

Separately, multiple testing and sequential testing are both quite mature fields, the former dating back to classical multiple comparison procedures of \citet{Fisher32}, \citet{Scheffe53}, Tukey, and others \citep[see][]{Seber03} for testing hypotheses about parameter vectors in linear models. Work on sequential hypothesis testing dates back to Wald's \citeyearpar{Wald47} invention of sequential analysis following World War~II; see \citet{Siegmund85} for a summary of the major developments.  However, the intersection of these two areas is less well-developed in a general setting.  One area that has been considered is the adaptation of some classical fixed sample size tests about vector parameters, such as those mentioned above, to the sequential sampling setting, including O'Brien and Fleming's~\citeyearpar{OBrien79} sequential version of Pearson's $\chi^2$ test, and Tang et al.'s \citeyearpar{Tang89,Tang93} group sequential extensions of O'Brien's \citeyearpar{OBrien84} generalized least squares statistic. For bivariate normal populations, \citet{Jennison93} proposed a sequential test of two one-sided hypotheses about the bivariate mean vector, and  \citet{Cook94b} proposed a sequential test in  a similar setting but where one of the hypotheses is two-sided. A procedure for comparing three treatments was proposed by \citet{Siegmund93}, related to Paulson's \citeyearpar{Paulson64} earlier procedure for selecting the largest mean of $k$ normal distributions, which \citet{Bartroff10e} showed to be a special case of their more general sequential step-down method; this procedure is presented in Section~\ref{sec:BL} where we give a simplified proof of its FWER control using our rejection principle.  Recently \citet{Ye13} proposed a group sequential Holm procedure that is also a special case of the procedure of \citet{Bartroff10e}, which allows arbitrary sampling schemes in addition to group sequential. The first sequential procedures to simultaneously control both the type~I and II FWERs were introduced by \citet{De12,De12b}.  \citet{Bartroff14b} propose a different approach to sequential control of both type~I and II FWERs, and their procedure will also be discussed in Section~\ref{sec:I&II} and shown to satisfy this rejection principle.

\section{A rejection principle for sequential tests}
\label{sec:theorem}
We present a general framework for testing multiple hypotheses with sequential data, i.e., with data streams. Assume that there are $k\ge 2$ data streams \begin{equation}\label{streams}
X_1^{(j)},X_2^{(j)},\ldots,\qmq{for}j=1,\ldots,k.
\end{equation} In general we make no assumptions about the dimension  of the sequentially-observed data $X_n^{(j)}$, which may themselves be vectors of varying size, nor about the dependence structure of within-stream data $X_n^{(j)}, X_{n'}^{(j)}$ or between-stream data $X_n^{(j)}, X_{n'}^{(j')}$ ($j\ne j'$). Assume that for each data stream~$j=1,\ldots,k$, there is a parameter vector~$\theta^{(j)}\in\Theta^{(j)}$ governing that stream $X_1^{(j)}, X_2^{(j)},\ldots$, and it is desired to test  a hypothesis $H^{(j)}\subseteq \Theta^{(j)}$ about $\theta^{(j)}$, with $H^{(j)}$ considered true if $\theta^{(j)}\in H^{(j)}$, and false otherwise.  The global parameter $\theta=(\theta^{(1)},\ldots,\theta^{(k)})$ is the concatenation of the individual parameters and is contained in the global parameter space $\Theta=\Theta^{(1)}\times\cdots\times \Theta^{(k)}$. Each $\theta\in\Theta$ indexes a probability measure $P_{\theta}$. With $\mH=\{H^{(1)},\ldots,H^{(k)}\}$ denoting the set of hypotheses to be tested, given $\theta=(\theta^{(1)},\ldots,\theta^{(k)})\in\Theta$ we let 
$$\mathcal{T}(\theta) = \{ H^{(j)} \in \mH: \theta^{(j)} \in H^{(j)}\}$$ denote the collection of true hypotheses when $P_{\theta}$ is the underlying probability measure, and 
\begin{equation}\label{F}
\mF(\theta) = \{ H^{(j)} \in \mH: \theta^{(j)} \notin H^{(j)}\}=\mH\setminus \mathcal{T}(\theta) 
\end{equation}
the false hypotheses. The familywise error rate is the probability of rejecting any true hypothesis,
\begin{equation}\label{FWEI}
\mbox{FWER}=\mbox{FWER}(\theta)=P_\theta(\mbox{any $H^{(j)}\in\mathcal{T}(\theta)$ rejected}).
\end{equation} In what follows, we will frequently drop the argument $\theta$ from these expressions for brevity.

At any point during sampling we shall refer to the \emph{active} hypotheses as those that have not yet been accepted or rejected, and \emph{active} data streams as those corresponding to active hypotheses.  A \emph{sequential multiple testing procedure} for the data streams~\eqref{streams} is simply a sampling and decision procedure that maps the current data from all the data streams and the list of active hypotheses to one of the following:
\begin{enumerate}[(a)]
\item A list of one or more active hypotheses to reject;
\item A list of one or more active hypotheses to accept;
\item\label{add.samp} An additional sample size to draw from each active data stream  before reevaluation.
\end{enumerate}
We note that the additional sample size in (\ref{add.samp}) can be 1, as in full sequential sampling. As a simplistic example, suppose there are $k=2$ data streams~\eqref{streams} and it is desired to test the respective hypotheses $H^{(1)}$ and $H^{(2)}$. A multiple testing procedure may first decide to sample 10 observations from the streams, yielding
\begin{align*}
&X_1^{(1)},X_2^{(1)}\ldots,X_{10}^{(1)}\\
&X_1^{(2)},X_2^{(2)}\ldots,X_{10}^{(2)}.
\end{align*} Based on this data the procedure may decide to reject $H^{(2)}$  and then sample a single additional observation from stream~1 (the lone remaining active stream), yielding
\begin{align*}
&X_1^{(1)},X_2^{(1)}\ldots,X_{10}^{(1)},X_{11}^{(1)}\\
&X_1^{(2)},X_2^{(2)}\ldots,X_{10}^{(2)}.
\end{align*} At this point the procedure may decide not to accept or reject $H^{(1)}$ but rather sample an additional 7 observations from stream~1, yielding
\begin{align*}
&X_1^{(1)},X_2^{(1)}\ldots,X_{10}^{(1)},  \ldots,X_{18}^{(1)}\\
&X_1^{(2)},X_2^{(2)}\ldots,X_{10}^{(2)},
\end{align*}  at which time the procedure may decide to accept $H^{(1)}$. Examples of sequential multiple testing procedures will be given in Section~\ref{sec:apps}, below.

Our main result, given in Theorem~\ref{thm:rej.princ}, extends a result of
\citet{Goeman10} for fixed sample size multiple testing procedures to the
sequential setting (i.e., for testing on data streams) in which sequential
sampling may occur between acceptances/rejections of hypotheses. Given any sequential
multiple testing procedure meeting the above general definition, its
rejection behavior (and hence its FWER, as we will see) can be described
by its \emph{rejection function} which we denote by $\rho$, which is a
possibly random function mapping the set $\mR \subseteq\mH$ of already rejected hypotheses, the set $\mA \subseteq\mH$ (disjoint from $\mR$) of already accepted hypotheses, the current sample size $n\in\mN$, and the set $\mD_n$ of all data available at time~$n$ to a set $\rho(\mR,\mA,n,\mD_n)\subseteq\mH\setminus(\mR\cup\mA)$ of hypotheses to reject. Here $\mN$ is the set
of all possible streamwise sample sizes of the procedure. Since the last argument~$\mD_n$ of $\rho(\mR,\mA,n,\mD_n)$ will always be the available data at time~$n$, in what follows we denote $\rho(\mR,\mA,n,\mD_n)$ simply by $\rho(\mR,\mA,n)$. Letting $\varnothing$ denote the empty set, the
value $\rho(\mR,\mA,n)= \varnothing$ indicates that either additional
sampling will be performed or that the testing procedure is terminated,
which occurs if $n=\max\mN$ or $\mR\cup\mA=\mH$. By convention define
$\rho(\mR, \mA,\infty) = \varnothing$. Therefore, iterations of $\rho$
describe the procedure's successive rejections of hypotheses,
and we will keep track of all the hypotheses that have been rejected after
each iteration of $\rho$ in sets $\mR_i\subseteq\mH$, and the hypotheses
that have been accepted after each iteration in $\mA_i\subseteq\mH$. To
this end, define $\mR_0 = \varnothing$, $n_0=0$, and
\begin{multline}\label{rej.func} \mR_{i+1} = \mR_{i} \cup \rho(\mR_{i},
\mA_{i},n_{i+1})\qmq{for $i=0,1,2,\dots$, where}\\ n_{i+1} =
\inf\left\{n \in \mN, n \geq n_{i} : \rho(\mR_{i}, \mA_{i},n) \neq
\varnothing\right\}, \end{multline} letting $\inf\varnothing=\infty$ as
usual. In what follows we do not need to define the sets~$\mA_i$ of
accepted hypotheses explicitly as we did for the $\mR_i$ in
\eqref{rej.func}, we only assume that the $\mA_i$ are defined in some way
such that $\varnothing=\mA_0\subseteq\mA_1\subseteq\ldots$ and $\mA_i\cap
\mR_i=\varnothing$ for all $i$. There are at most $k$ nontrivial
iterations of $\rho$ in the sense that $\rho\ne \varnothing$, by virtue of
the fact that there are $k$ hypotheses and hence at most $k$ hypotheses
that could be rejected.  Consequently, $\mR_k=\mR_{k+1}=\ldots$ and this
common set is the totality of all hypotheses rejected by the procedure,
and FWER$(\theta)$ can thus be written $P_\theta(\mR_k \not\subseteq
\mF(\theta))$.

The following theorem gives a rejection principle for sequential tests and establishes its sufficiency for FWER control.

\begin{theorem}\label{thm:rej.princ}
Let $\theta \in \Theta$ denote the true value of the global parameter,  $\alpha\in(0,1)$, and $\mH$ and $\mR_k$ as defined above. If $\rho$ and $\mN$ are the rejection function and sample size set, respectively, of a sequential multiple testing procedure such that
\begin{enumerate}
\item for any subsets $\mR, \mR', \mA$ of $\mH$ with $\mR \subseteq \mR'$ and $\mA\cap \mR=\varnothing$, and any $n \in \mN$, we have \begin{equation}
\label{monotonicity}
\rho(\mR, \mA, n) \subseteq \rho(\mR', \varnothing, n) \cup \mR'
\end{equation}
with $P_\theta$-probability 1, and 
\item
\begin{equation}
\label{singlestep}
P_\theta(\rho(\mF(\theta),\varnothing, n) \subseteq \mF(\theta) \text{ for all } n \in \mN) \geq 1-\alpha,
\end{equation}
\end{enumerate}
then \begin{equation}\label{princ.FWE}
P_\theta(\mR_k \not\subseteq \mF(\theta)) \le\alpha,
\end{equation}i.e., FWER$(\theta)$ is no greater than $\alpha$.
\end{theorem}
\begin{proof}
Let $\mF=\mF(\theta)$, $V=\{\rho(\mF,\varnothing,n) \subseteq \mF \text{ for all } n \in \mN\}$, and $W_i=\{\mR_i\subseteq\mF\}$, $i=0,1,\ldots$. We will prove by induction that, as events, $V\subseteq W_i$ for all $i\ge 0$; the result~\eqref{princ.FWE} then follows from the $i=k$ case and \eqref{singlestep}.  The $i=0$ case is trivial since $\mR_0=\varnothing \subseteq \mF$. Suppose that $V\subseteq W_i$.  Then on $V$ we have 
\begin{align*}
\mR_{i+1} &= \mR_i\cup \rho(\mR_i, \mA_i,n_{i+1})\qm{[by \eqref{rej.func}]}\\
&\subseteq \mR_i\cup \left(\rho(\mF, \varnothing,n_{i+1}) \cup \mF\right)\qm{[by \eqref{monotonicity} and the inductive hypothesis]}\\
&=\rho(\mF, \varnothing,n_{i+1}) \cup \mF \qm{[by the inductive hypothesis]}\\
&= \mF.
\end{align*}
\end{proof}

The rejection principle for fixed sample size procedures presented in \citet{Goeman10} can be regarded as a special case of Theorem~\ref{thm:rej.princ}, since all fixed sample size procedures are sequential procedures with the fixed sample size being the only element of $\mN$.

\section{Applications of this rejection principle}\label{sec:apps}
In this section we apply the rejection principle in 
Theorem~\ref{thm:rej.princ}  in a number of settings, some with special structure and some without.

\subsection{Testing hypotheses without assumed special structure}\label{sec:no.struc}
\subsubsection{A sequential step-down procedure}\label{sec:BL}
\citet{Bartroff10e} proposed a sequential multiple testing procedure, extending  Holm's~\citeyearpar{Holm79} fixed sample size step-down procedure, that controls FWER regardless of between-stream dependence, requiring only that each hypothesis have a sequential test statistic which marginally controls the conventional type~I error probability. After briefly introducing Bartroff and Lai's procedure, we show that its error control is a special case of Theorem~\ref{thm:rej.princ}.

Here we present the Bartroff-Lai procedure in slightly more generality than in their original paper.  In particular, here we remove the need for (a) common critical values among the $k$ sequential test statistics by using standardizing functions, below, introduced by \citet{Bartroff14b},
and (b) critical values for all possible significance levels; here we only need critical values corresponding to certain fractions of the desired FWER bound~$\alpha$. Given a set~$\mN$ of possible per-stream sample sizes, assume that, for each $j=1,\ldots,k$, associated with the $j$th hypothesis~$H^{(j)}$ and data stream~$X_1^{(j)}, X_2^{(j)},\ldots$ is a scalar-valued sequential test statistic~$T_n^{(j)}=T_n^{(j)}(X_1^{(j)}, \ldots,X_n^{(j)})$ with $k$ critical values $B_1^{(j)}\ge\ldots\ge B_k^{(j)}$ such that 
\begin{equation}
\label{criticalvalue}
P_{\theta^{(j)}}\left( T_n^{(j)}\ge B_s^{(j)}\qm{for some $n \in \mN$}\right)\le\frac{\alpha}{k-s+1}\qmq{for all} \theta^{(j)} \in H^{(j)},
\end{equation} for all $s=1,\ldots,k$. The inequality \eqref{criticalvalue} just says that the  sequential test that stops and rejects $H^{(j)}$ at the first $n\in \mN$ such that $T_n^{(j)}\ge B_s^{(j)}$, and accepts $H^{(j)}$ otherwise, has type~I error probability $\alpha/(k-s+1)$.  For $j=1,\ldots,k$ define the \emph{standardizing function} 
\begin{equation}\label{std.func}
\vphi^{(j)}(x)=\begin{cases}
x-B_k^{(j)}+1,&\mbox{for}\;x\le B_k^{(j)}\\
\frac{x-B_s^{(j)}}{B_{s-1}^{(j)}-B_s^{(j)}}+k-s+1,&\mbox{for}\;B_s^{(j)}\le x\le B_{s-1}^{(j)}\;\mbox{if}\;B_{s-1}^{(j)}>B_s^{(j)},\quad 1<s\le k\\
x-B_1^{(j)}+k,&\mbox{for}\;x\ge B_1^{(j)},
\end{cases}
\end{equation}
which is an increasing, piecewise-linear  function such that $\vphi^{(j)}(B_s^{(j)})=k-s+1$ for $s=1,\ldots,k$, and thus
\begin{equation}\label{Tiffphi}
T_n^{(j)}\ge B_s^{(j)}\quad\Leftrightarrow\quad \vphi^{(j)}(T_n^{(j)})\ge k-s+1.
\end{equation}
 The standardizing functions will be applied to the test statistics before ranking them and they allow us to compare the test statistics $T_n^{(1)},\ldots,T_n^{(k)}$, which may be on different scales. In general, the standardizing function can be any increasing function such that $\vphi^{(j)}(B_s^{(j)})$ does not depend on $j$. To use a different standardizing function, all that would need to be adjusted in what follows is the right hand side of the inequality in \eqref{BL.mi}, below.

Letting $\mI_1=\{1, 2, \dots, k\}$, $r_1=0$, and $n_0=0$, the $i$th stage ($i=1, \dots, k$) of the Bartroff-Lai procedure proceeds as follows. 
\begin{enumerate}
\item\label{BL.step1} Sample each active data stream $\{X_n^{(j)}\}_{j\in \mI_i}$ up to sample size
$$n_i=\inf\left\{n \in \mN : n > n_{i-1}\qmq{and} T_n^{(j)} \ge B_{r_i+1}^{(j)} \qmq{for some} j\in \mI_i\right\}.$$
\item With $\vphi^{(j)}$ given by \eqref{std.func}, standardize and order the active test statistics $\wtilde{T}^{(j)}_{n_i}=\vphi^{(j)}(T^{(j)}_{n_i})$, $j\in\mI_i$, as follows:
$$\wtilde{T}_{n_i}^{(j(i,1))}\ge \wtilde{T}_{n_i}^{(j(i,2))}\ge\ldots\ge \wtilde{T}_{n_i}^{(j(i,|\mI_i|))}.$$
\item Reject $H^{(j(i,1))}, H^{(j(i,2))},\ldots, H^{(j(i,m_i))}$, where
\begin{equation}\label{BL.mi}
m_i = \min\left\{m\geq 1: \wtilde{T}_{n_i}^{(j(i,m+1))}< k-r_i-m \right\}.
\end{equation}
\item\label{BL.step4} If $i=k$ or $n_i =\max \mN$, stop and accept all remaining active hypotheses. Otherwise, let $\mI_{i+1}$ be the indices of the remaining hypotheses, set $r_{i+1}=r_i+m_i$, and continue on to stage $i+1$.
\end{enumerate}
An important caveat is that, at any point, any of the active hypotheses may be accepted without violating the FWER control proved below, as long as the set~$\mI_i$ of active hypotheses  is appropriately updated.  To maintain generality, here we do not specify an acceptance rule for the Bartroff-Lai procedure.  Sequential multiple testing procedures with explicit acceptance, as well as rejection, rules are considered in Section~\ref{sec:I&II}, which control both the type~I and II FWERs, the latter defined there.

Because the Bartroff-Lai procedure presented here is slightly more general than the one proved to control FWER in \citet[][Theorem~2.1]{Bartroff10e}, we record this procedure's FWER control in Corollary~\ref{thm:BL} which we prove by applying the rejection principle in Theorem~\ref{thm:rej.princ}.

\begin{corollary}\label{thm:BL}
If \eqref{criticalvalue} holds then the procedure defined above in steps~\ref{BL.step1}-\ref{BL.step4} satisfies FWER$(\theta)\le\alpha$ for all $\theta\in\Theta$.
\end{corollary}

\begin{proof}It is not hard to see that the rejection function of the above procedure is given by
\begin{equation}\label{rho.BL}
\rho(\mR,\mA,n)=\left\{H^{(j)}\in\mH\setminus(\mR\cup\mA):\quad \wtilde{T}_n^{(j)}\ge k-|\mR| \right\},
\end{equation} about which we verify \eqref{monotonicity} and \eqref{singlestep}. For \eqref{monotonicity}, given $\mR, \mR', \mA$ as described there and $n\in\mN$, if  $H^{(j)}\in\rho(\mR,\mA,n)\setminus\mR'$ then $\wtilde{T}_n^{(j)}\ge k-|\mR|\ge k-|\mR'|$ since $\mR\subseteq\mR'$,  hence $H^{(j)}\in\rho(\mR',\varnothing,n)$ so $\rho$ satisfies \eqref{monotonicity}. For \eqref{singlestep}, without loss of generality assume that $\mT\ne\varnothing$ since the following probability is zero otherwise. Let $V_j=\{\wtilde{T}_n^{(j)}\ge k-|\mF|\;\mbox{for some $n\in\mN$} \}$. Using the Bonferroni inequality, \eqref{Tiffphi}, and \eqref{criticalvalue},
\begin{multline*}
P_\theta(\rho(\mF,\varnothing,n)\not\subseteq\mF\qm{for some $n\in\mN$})
=P_\theta\left(\bigcup_{j:\,H^{(j)}\in\mT}V_j\right) \le \sum_{j:\,H^{(j)}\in\mT}P_{\theta^{(j)}}\left(V_j\right)\\
 = \sum_{j:\,H^{(j)}\in\mT}P_{\theta^{(j)}}\left(T_n^{(j)}\ge B_{|\mF|+1}^{(j)}\;\mbox{for some $n\in\mN$}\right) \le \sum_{j:\,H^{(j)}\in\mT}\frac{\alpha}{k-|\mF|}  =|\mT|\cdot\frac{\alpha}{|\mT|}
=\alpha.
\end{multline*}
\end{proof}

\subsubsection{Tests that simultaneously control type~I and II FWERs}\label{sec:I&II}

Extending the procedure in the previous section, \citet{Bartroff14b} proposed a sequential test that simultaneously controls both the type~I and II FWERs, the latter defined below in \eqref{FWEII} analogously to the type~I version~\eqref{FWEI}.  The error control of this procedure can also be seen as a special case of the rejection principle. Adding to the setup in Section~\ref{sec:theorem}, suppose one also has alternative hypotheses $G^{(1)},\ldots,G^{(k)}$ such that $G^{(j)}\subseteq\Theta^{(j)}$ and $G^{(j)}\cap H^{(j)}=\varnothing$ for all $j=1,\ldots,k$. With this we redefine the false hypotheses from \eqref{F} to be 
\begin{equation*}
\mF(\theta) = \{ H^{(j)} \in \mH: \theta^{(j)} \in G^{(j)}\}
\end{equation*} and define the type~II FWER as 
\begin{equation}\label{FWEII}
\mbox{FWER}_{II}(\theta)=P_\theta(\mbox{any $H^{(j)}\in\mathcal{F}(\theta)$ accepted}).
\end{equation} 

Given desired FWER bounds $\alpha$ and $\beta$, the procedure requires only that each data stream $X_1^{(j)}, X_2^{(j)},\ldots$ has a scalar-valued sequential test statistic~$T_n^{(j)}=T_n^{(j)}(X_1^{(j)}, \ldots,X_n^{(j)})$ with critical values $A_1^{(j)},\ldots, A_k^{(j)}, B_1^{(j)}, \ldots, B_k^{(j)}$ such that
\begin{align}
P_{\theta^{(j)}}(T_n^{(j)}\ge B_s^{(j)}\;\mbox{some $n$,}\; T_{n'}^{(j)}>A_1^{(j)}\;\mbox{all $n'<n$})&\le \frac{\alpha}{k-s+1}\qmq{for all}\theta^{(j)}\in H^{(j)}\label{typeI}\\
 P_{\theta^{(j)}}(T_n^{(j)}\le A_s^{(j)}\;\mbox{some $n$,}\; T_{n'}^{(j)}<B_1^{(j)}\;\mbox{all $n'<n$})&\le \frac{\beta}{k-s+1}\qmq{for all}\theta^{(j)}\in G^{(j)}\label{typeII}
\end{align} for all $j, s=1,\ldots,k$. These inequalities simply guarantee that each sequential test marginally controls the conventional type~I and II error probabilities at desired fractions of $\alpha, \beta$.

For brevity we do not restate Bartroff and Song's
\citeyearpar{Bartroff14b} procedure here, but rather just say that it 
has a similar flavor to the one in Section~\ref{sec:BL} but is more 
complex in that it interweaves rejections
and acceptances of the $H^{(j)}$ at each stage. It also utilizes a
standardizing function, mapping $B_s^{(j)}$ to $k-s+1$ as above in
\eqref{std.func}, and mapping $A_s^{(j)}$ to $-(k-s+1)$. The procedure
controls the type~I and II FWERs, regardless of dependence between the
data streams, as long as \eqref{typeI}-\eqref{typeII} hold. This can be
easily proved using the rejection principle, whose application
here is interesting because it is used to prove control of type~II FWER as
well as type~I. The proof proceeds by defining the procedure's \textit{acceptance
function} $\wtilde{\rho}(\mR,\mA,n)$, analogous to the rejection
function~$\rho$ in Section~\ref{sec:theorem}, and these two are alternated
to give the procedure's accept/reject decisions.  Then
Theorem~\ref{thm:rej.princ} is applied to both $\rho$ and $\wtilde{\rho}$
separately to prove type~I and II FWER control, respectively. For this
procedure, $\rho$ takes the same form~\eqref{rho.BL} and the acceptance
function is similar, \begin{equation*}
\wtilde{\rho}(\mR,\mA,n)=\left\{H^{(j)}\in\mH\setminus(\mR\cup\mA):\quad
\wtilde{T}_n^{(j)}\le -(k-|\mA|) \right\}. \end{equation*}

\subsection{Testing hypotheses with special structure}\label{sec:struct}
Whereas  the previous sections assumed no special structure of the hypotheses being tested, in some settings logical relationships or priorities exist among the hypotheses which can be exploited by testing the hypotheses in a certain order and allow less stringent (i.e., more powerful) tests to be used. In this section we consider sequentially testing hypotheses in order \citep{Rosenbaum08}, and later the  special case of sequentially testing closed hypotheses \citep{Marcus76}.

\subsubsection{Testing hypotheses in order}\label{sec:IO} In many multiple
testing situations it is natural to only test a certain hypothesis if
certain other hypotheses have already been rejected; two real examples are
given in Section~\ref{sec:ex}. \citet{Rosenbaum08} considered various
ordering schemes and gave fixed sample size tests which control the FWER.  The
most general ordering scheme \citet{Rosenbaum08} considers is the
following, although his results apply to hypotheses and partitions with
more general (e.g., infinite) index sets, whereas here we simply consider
hypotheses indexed by $\{1,\ldots,k\}$ for coherence with the previous
sections. Let $\mH_1,\ldots,\mH_s$ be a partition of $\mH$ such that it is
desired to only test the hypotheses in $\mH_i$ if all the hypotheses in
$\bigcup_{i'<i}\mH_{i'}$ have already been rejected. Recall that
$\mH_1,\ldots,\mH_s$ being a partition of $\mH$ means that the $\mH_i$ are disjoint and their union is $\mH$. For $j=1,\ldots,k$, let $i_j$ denote the unique index $i$ of the
$\mH_i$ containing $H^{(j)}$, i.e., $H^{(j)}\in\mH_{i_j}$.
\citet{Rosenbaum08} calls a subset~$\mH'\subseteq\mH$ \textit{exclusive}
if at most one hypothesis $H^{(j)}\in\mH'$ is true, and $\mH_1, \cdots,
\mH_s$ is \textit{sequentially\footnote[3]{Here, to be faithful to Rosenbaum's \citeyearpar{Rosenbaum08} terminology, we slightly abuse our stated convention by allowing the ``sequential'' in sequentially exclusive to refer to this stepwise condition on the partition, not the sequential nature of the data.} exclusive} if all the hypotheses in
$\bigcup_{i'<i}\mH_{i'}$ being false implies that $\mH_i$ is exclusive,
for all $i=1,\ldots,s$. In the fixed sample size setting with valid $p$-values
$p^{(1)},\ldots,p^{(k)}$ for testing $H^{(1)},\ldots,H^{(k)}$,
respectively, \citet[][Proposition~3]{Rosenbaum08} shows\footnote[4]{Rosenbaum's \citeyearpar{Rosenbaum08} requirement that
the $\mH_i$ be ``intervals'' is not needed if one does not require that
the hypotheses $H^{(1)},\ldots, H^{(k)}$ be strictly tested in the indexed
order, since the hypotheses can simply be re-indexed within each subset
$\mH_i$ to make it an interval.} that if $\mH_1, \cdots, \mH_s$ are
sequentially exclusive, then the following test controls the FWER at level
$\alpha$: Reject $H^{(j)}$ if and only if $p^{(j)}\le\alpha$ and all
hypotheses in $\bigcup_{i<i_j}\mH_{i}$ have already been rejected.

Here we present a sequential multiple testing procedure for testing hypotheses in order, and use the rejection principle in Theorem~\ref{thm:rej.princ} to prove its FWER control. We adopt the notation for data streams, parameters, and hypotheses given in Section~\ref{sec:theorem}, and we assume that there is a sequentially exclusive partition $\mH_1, \cdots, \mH_s$ of $\mH$ representing the desired order of testing. Given a desired  FWER bound~$\alpha$ and a set~$\mN$ of possible streamwise sample sizes, we also assume that, for each $j=1,\ldots,k$, associated with the data stream~$X_1^{(j)}, X_2^{(j)},\ldots$ and hypothesis~$H^{(j)}$ is a scalar-valued sequential test statistic~$T_n^{(j)}=T_n^{(j)}(X_1^{(j)}, \ldots,X_n^{(j)})$ with a critical value $B^{(j)}$ satisfying 
\begin{equation}\label{typeI.IO}
P_{\theta^{(j)}}\left( T_n^{(j)}\ge B^{(j)}\qm{for some $n \in \mN$}\right)\le \alpha\qmq{for all}\theta^{(j)} \in H^{(j)}.
\end{equation} Note that here we only need a single critical value~$B^{(j)}$ for each test statistic rather than the $k$ critical values needed in the more general, unstructured setup of Section~\ref{sec:BL}, which our exploitation of the sequential exclusivity property here will allow us to sidestep.

Let $\mI_1=\{1,\ldots,k\}$, $\ell_1=1$ and $n_0=0$. The $i$th stage ($i=1,\ldots,k$) of the sequential procedure for testing hypotheses in order proceeds as follows.
\begin{enumerate}
\item\label{IO.step1} Sample each active data stream $\{X_n^{(j)}\}_{j\in \mI_i}$ up to sample size
$$n_i=\inf\left\{n \in \mN : n > n_{i-1} \qmq{and}T_n^{(j)} \ge B^{(j)} \qmq{for some} j: H^{(j)}\in \mH_{\ell_i}\right\}.
$$
\item Reject $H^{(j)} \in \mH_{\ell_i}$ if all hypotheses in $\bigcup_{\ell'<\ell_i}\mH_{\ell'}$ have been rejected and $T_{n_i}^{(j)}\ge B^{(j)}$.
\item\label{IO.step3} If $i=k$ or $n_i =\max \mN$, stop and accept all remaining hypotheses. Otherwise, set
\begin{align*}
\ell_{i+1}&=\begin{cases}
\ell_i+1, & \mbox{if all hypotheses in $\mH_{\ell_i}$ have been rejected,}\\
\ell_i, &\mbox{otherwise},
\end{cases}\\
\mI_{i+1}&=\left\{j: H^{(j)}\in \bigcup_{\ell'\ge \ell_{i+1}}\mH_{\ell'}\qm{and $H^{(j)}$ has not been rejected}\right\},
\end{align*} and continue on to stage $i+1$.
\end{enumerate}
The FWER control of this procedure is easily established using the rejection principle of Theorem~\ref{thm:rej.princ}.
\begin{corollary}
\label{order}
If $\mH_1,\ldots,\mH_s$ is a sequentially exclusive partition of $\mH$ and \eqref{typeI.IO} holds, then the procedure defined above in steps~\ref{IO.step1}-\ref{IO.step3} satisfies FWER$(\theta)\le\alpha$ for all $\theta\in\Theta$.
\end{corollary}
\begin{proof} The rejection function is given by
\begin{equation*}\label{rej.func.IO}
\rho(\mR,\mA,n)=\{H^{(j)}\in\mH\setminus(\mR\cup\mA): T_n^{(j)} \ge  B^{(j)}\qmq{and}\mH_\ell\subseteq\mR\qmq{for all} \ell<i_j\},
\end{equation*}
about which we verify \eqref{monotonicity} and \eqref{singlestep}.
For \eqref{monotonicity}, with $\mR,\mR',\mA$ as described there and $n\in \mN$, if $H^{(j)}\in\rho(\mR,\mA,n)\setminus\mR'$ then all conditions for $H^{(j)}$ to be in $\rho(\mR',\varnothing,n)$ are satisfied, the latter since $\mH_\ell\subseteq\mR\subseteq \mR'$ for all $\ell<i_j$. For \eqref{singlestep}, without loss of generality assume $\mT\ne\varnothing$ and let $\ell^*$ be the smallest index of a subset $\mH_{\ell^*}$ containing a true hypothesis, i.e., $\ell^*=\min\{\ell: \mH_\ell\cap\mT\ne\varnothing\}$, and let $H^{(j^*)}$ be an arbitrarily chosen but fixed hypothesis in $\mH_{\ell^*}\cap\mT$. On $V:=\{\rho(\mF,\varnothing,n)\not\subseteq\mF \;\mbox{for some}\;n\in \mN\}$, there is some true $H^{(j)}\in\rho(\mF,\varnothing,n)$ with $T_n^{(j)}\ge B^{(j)}$ and $\mH_\ell\subseteq\mF$ for all $\ell<i_j$. It follows from the latter that $i_j=\ell^*$ and by this and sequential exclusivity, $j=j^*$. Using these facts and \eqref{typeI.IO} we have
$$P_\theta(V)\le P_{\theta^{(j^*)}}\left(T_n^{(j^*)}\ge B^{(j^*)}\;\mbox{for some $n\in \mN$}\right)\le\alpha,$$ showing that $\rho$ satisfies \eqref{singlestep}.
\end{proof}

\subsubsection{Closed Testing}\label{sec:closed}
A frequently encountered special case of testing hypotheses in order is 
closed testing.  The set of hypotheses $\mH=\{H^{(1)},\ldots,H^{(k)}\}$ is 
\emph{closed} if it is closed under intersection.  \citet{Marcus76} 
introduced a fixed sample size method of testing a closed set~$\mH$ that 
controls the FWER and only requires a level-$\alpha$ test of each 
intersection hypothesis $\bigcap_{j\in J} H^{(j)}$, $J\subseteq 
\{1,\ldots,k\}$. Beginning with the \emph{global hypothesis} $\bigcap_{j=1}^k H^{(j)}$, their procedure tests the elements of $\mH$ in order of decreasing \textit{dimension} (the maximum number of~$H^{(j)}$ being intersected), and $H\in\mH$ is tested if and only if all elements of $\mH$ contained in $H$ have been rejected.  Fixed sample size closed testing is a special case of Rosenbaum's \citeyearpar{Rosenbaum08} testing in order formulation.

In the sequential realm, \citet{Tang99} gave a group sequential procedure for closed testing of hypotheses about multivariate normal data.  A more general sequential procedure for closed testing can be derived using the rejection principle via the sequential procedure in Section~\ref{sec:IO} and Corollary~\ref{order}. The relevant partition of $\mH$ is the following, defined inductively for $i=1,\ldots,k$:
\begin{equation}\label{IO.part}
\mH_i=\left\{H=\bigcap_{j\in J} H^{(j)}: |J|=k-i+1,\quad H\not\in \mH_{i'}\qm{any $i'<i$}\right\}.
\end{equation} 
The subset $\mH_i$ contains all hypotheses of dimension $k-i+1$, as is guaranteed by the last condition in \eqref{IO.part}. For example, $\mH_1$ contains only the global hypothesis, $\mH_2$ contains all intersections of dimension $k-1$, and so on.  Applying the sequential procedure in Section~\ref{sec:IO} to this partition results in a sequential procedure that tests the hypotheses in order of decreasing dimension, using a level-$\alpha$ test for each, with sampling of the active data streams occurring between rejection decisions. After establishing that the partition~\eqref{IO.part} is sequentially exclusive, it follows immediately from Corollary~\ref{order} that this procedure controls the FWER.

\begin{corollary}\label{cor:closed}
If \eqref{typeI.IO} holds and $\mH$ is closed, then the partition~\eqref{IO.part} is sequentially exclusive, hence the procedure defined in steps~\ref{IO.step1}-\ref{IO.step3} of Section~\ref{sec:IO} applied to \eqref{IO.part} satisfies FWER$(\theta)\le\alpha$ for all $\theta\in\Theta$.
\end{corollary}

\begin{proof} To establish sequential exclusivity, suppose there are distinct hypotheses $H, H'\in\mH_i$ that are both true, i.e., $\theta\in H$ and $\theta\in H'$. Then $\theta\in H\cap H'$ so $H\cap H'$ is true, and $H\cap H'\in\mH_{i'}$ for some $i'<i$ by virtue of $H, H'$ being distinct.
\end{proof}

A similar but distinct formulation of sequentially testing closed hypotheses is given in \citet[][Theorem~2.2]{Bartroff10e}, which does not explicitly force testing in order of decreasing dimension, but rather gives a sufficient condition on the test statistics under which the procedure in Section~\ref{sec:BL} would test in this order anyway.

\section{Examples}\label{sec:ex}
In this section we give two examples of the sequential procedures in Section~\ref{sec:struct} applied to real testing situations.
In Section~\ref{sec:IOex} we apply the sequential procedure for testing in order to an observational study involving chromosome aberration data, and in Section~\ref{sec:closed.ex} we apply the sequential closed testing procedure to estimate the maximum safe dose of a treatment. In both cases the performance of the sequential procedure is compared with the corresponding fixed sample size procedure and the efficiency gain, in terms of savings in  average sample size, of the sequential procedure is highlighted.

\subsection{Chromosome aberrations of patients exposed to anti-tuberculosis drugs}\label{sec:IOex}

In non-randomized testing situations, such as observational studies, it is common for treatment responses to be compared with more than one ``control'' response, such as baseline and non-treatment, since this can provide information on differences due to nonrandom treatment assignment \citep[e.g., see][Section~8]{Rosenbaum02}.

\citet{Masjedi00} studied possible mutagenic effects of anti-tuberculosis drugs by comparing the frequency of chromosome aberrations including gaps per 100 cells in an observational study of $n=36$ patients before (denoted $b$) and after (denoted $a$) the treatment, and 36 healthy controls (denoted $c$), who matched the treatment group by sex and age and were selected from relatives of the treatment group when possible. The response triples $(y_{ci},y_{bi},y_{ai})$, $i=1,\ldots,n$, are given in Table~\ref{table 1}, where larger numbers indicate more chromosome damage. 

\begin{table}[htdp]
\caption{\citet{Masjedi00} data on total chromosome aberrations per 100 cells including gaps. }
\begin{center}
\begin{tabular}{|c|c|c|}
\hline
Control~$y_{ci}$&Before treatment~$y_{bi}$&After treatment~$y_{ai}$\\
\hline
1.00&0.50&3.00\\
1.50&4.50&5.50\\
0.50&3.50&5.00\\
0.50&2,66&3.33\\
0.66&1.50&4.50\\
1.00&5.00&7.00\\
1.00&1.33&5.33\\
0.66&1.50&2.50\\
0.00&2.00&5.33\\
1.33&1.50&3.00\\
1.50&1.33&3.33\\
2.00&2.00&2.00\\
1.33&2.00&4.66\\
0.00&2.66&10.00\\
3.00&1.33&3.33\\
0.50&3.50&5.00\\
0.66&3.00&5.00\\
1.33&2.66&3.33\\
3.00&0.00&4.00\\
0.66&1.50&7.00\\
0.50&1.00&3.00\\
0.66&4.00&4.00\\
2.00&1.33&2.66\\
1.33&0.66&3.33\\
0.00&1.50&3.50\\
1.00&0.66&2.00\\
0.50&2.00&3.33\\
1.33&1.00&3.50\\
0.50&1.33&2.66\\
1.00&2.00&2.00\\
0.66&2.00&4.00\\
1.50&1.50&1.50\\
2.66&2.00&3.50\\
1.33&0.66&3.33\\
0.66&0.00&2.66\\
1.00&1.50&1.50\\
\hline
\end{tabular}
\end{center}
\label{table 1}
\end{table}

\citet{Rosenbaum08} considers the model
\begin{align*}
y_{ai}&=\mu_a+\pi_i+\lambda_i+\eps_i,\\
y_{bi}&=\mu_b+\pi_i+\lambda_i+\zeta_i,\\
y_{ci}&=\mu_c+\pi_i+\eta_i,
\end{align*}
$i=1, \dots, n$, where $\pi_i$, $\lambda_i$, $\eps_i$, $\zeta_i$, and $\eta_i$ are independent with continuous distributions $F_\pi$, $F_\lambda$, $F_\eps$, $F_\zeta$, and $F_\eta$, assumed to be symmetric about zero. Here $\mu$ represents the group effect, the random variables $\pi_i$ and $\lambda_i$ reflect the correlation due to matching and treatment vs.\ control, respectively, and $\eps_i$, $\zeta_i$ and $\eta_i$ are error terms. The primary scientific question was to ascertain whether the post-treatment responses exceed the baseline and control responses by more than the baseline and control responses differ from each other, i.e.,
\begin{equation*}
\mu_a-\max(\mu_b, \mu_c) > \max(\mu_b, \mu_c)-\min(\mu_b, \mu_c),
\end{equation*} hence the null hypothesis of primary interest is the negation of this,
\begin{equation*}
H_0: \mu_a-\max(\mu_b, \mu_c) \leq \max(\mu_b, \mu_c)-\min(\mu_b, \mu_c).
\end{equation*} However, even if $H_0$ cannot be rejected, it may be beneficial to be able to draw some weaker conclusions about the treatment. Namely,  \citet{Rosenbaum08} defines the five additional hypotheses:
\begin{align*}
&H_+: \mu_a \leq (\mu_b+\mu_c)/2,\\
&H_b: \mu_a \leq \mu_b,\\
&H_c: \mu_a \leq \mu_c,\\
&H_*: \mu_a-\mu_c \leq \mu_c-\mu_b,\\
&H_\sharp: \mu_a-\mu_b \leq \mu_b-\mu_c.
\end{align*}

The logical implications of these hypotheses suggest the sequentially exclusive partition $\mH_1=\{H_+\}$, $\mH_2=\{H_b, H_c\}$, $\mH_3=\{H_*, H_\sharp\}$, and $\mH_4=\{H_0\}$ of $\mH=\{H_0, H_+, H_b, H_c, H_*, H_\sharp\}$, and Figure~\ref{flowchart} describes the procedure for testing these hypotheses in order, which applies to both the fixed sample size procedure of \citet{Rosenbaum08} as well as the sequential procedure defined above  in which sampling may occur between decisions and Corollary~\ref{order} shows that the FWER is controlled. 

\begin{figure}[htbp]
\begin{center}
\small
{
\begin{tikzpicture}[node distance = 2cm, auto]
    % Place nodes
    \node [block](start) {Start};
    \node [block, below of=start] (first) {$H_+$ rejected at level $\alpha$?};
    \node [left of=first, node distance=4cm]{First hierarchy $\mH_1$};
    \node [block, below of=first] (second) {$H_b$ and $H_c$ \\rejected at level $\alpha$?};
    \node [left of=second, node distance=4cm]{Second hierarchy $\mH_2$};
    \node [block, right of=second, node distance=6cm] (stop) {Stop\\$H_0$ is not rejected};
    \node [block, below of=second] (third) {$H_*$ and $H_\sharp$ \\rejected at level $\alpha$?};
    \node [left of=third, node distance=4cm]{Third hierarchy $\mH_3$};
    \node [block, below of=third](final){$H_0$ is rejected};
    \node [left of=final, node distance=4cm]{Fourth hierarchy $\mH_4$};
       % Draw edges
    \path [line] (start) -- (first);
    \path [line] (first) -- node{no}(stop);
    \path [line](first)--node{yes}(second);
    \path [line] (second) -- node{no}(stop);
    \path [line](second)--node{yes}(third);
    \path [line](third)--node{no}(stop);
    \path [line](third)--node{yes}(final);
   \end{tikzpicture}
}
\end{center}
\caption{Flow chart of the testing procedure for the chromosome aberration example.\label{flowchart}}
\end{figure}
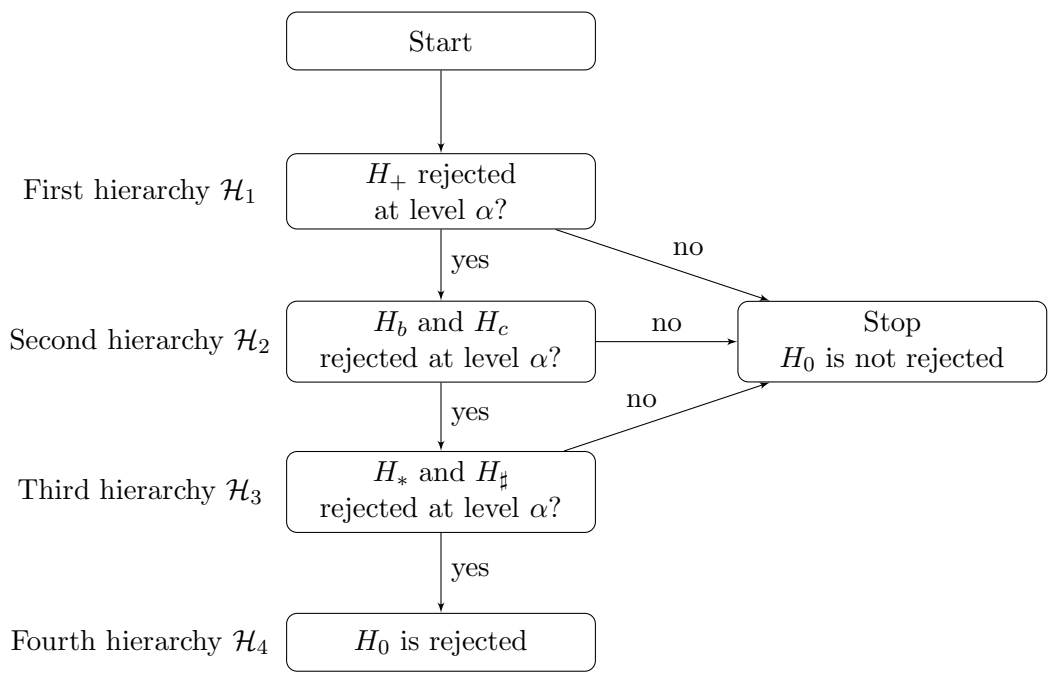

Because the \citet{Masjedi00} data $(y_{ai}, y_{bi}, y_{ci})$ exhibits strong non-normality, \citet{Rosenbaum08} applied the one-sided version of Wilcoxon's signed-rank test to obtain the fixed sample size  $p$-values, using $y_{ai}-(y_{bi}+y_{ci})/2$ to test $H_+$, $y_{ai}-y_{bi}$ to test $H_b$, and so on. Following this approach, we apply a  sequential version of these tests in which the sequential test statistics $T_n^{(j)}$ ($j=1,\ldots,6$) are taken to be the repeatedly-calculated $p$-values for Wilcoxon's test, with the common critical value $B=B^{(1)}=\ldots=B^{(6)}$ adjusted to control the marginal type~I error probability~\eqref{typeI.IO} at $\alpha=.05$, and was determined by Monte Carlo.  These sequential tests were then applied to the Masjedi et al.\ data in Table~\ref{table 1}, after setting aside the 5 tied observations for a total of 31 triples.  To assess the performance of this sequential test, Monte Carlo simulations were performed in which the order of the triples were permuted 50,000 times and the sample size needed to reach an accept/reject decision on all hypotheses was recorded each time.  In order to see the effect of different sequential sampling schemes, this was carried out for the five different choices of sample size sets~$\mN$ given in Table~\ref{table 3}: fully-sequential sampling in which $\mN=\{1,\ldots,31\}$, and four group-sequential schemes with 5, 4, 3, and 2 groups, respectively, with evenly sized groups (until the final group, which is 1 larger). Table~\ref{table 3} also reports the common critical value $B$ and the average total sample size over the 50,000 simulated paths for each scheme. In all cases, the average sample size was dramatically reduced from 31. Even in the worst case
of 2-stage sampling, with $\mN=\{15,31\}$, more than 9 observations were saved on average, nearly one third of the largest possible sample size, $31$. Interestingly, the 5-group scheme with $\mN=\{10,15,20,25,31\}$ has nearly the same average sample size as the fully-sequential scheme, which may be appealing in applications where fully sequential testing is not practical.

\begin{table}[htdp]
\caption{The sample size set~$\mN$, critical value~$B$, and average sample size of various Wilcoxon signed-rank tests for testing hypotheses about the chromosome aberration data.}
\begin{center}
\begin{tabular}{|c|c|c|}
\hline
$\mN$&$B$&Average sample size\\
\hline
$\{1, \dots, 31\}$&0.0031&18.9\\
$\{10,15,20,25,31\}$&0.0068&18.9\\
$\{10,20,25,31\}$&$0.0082$&20.1\\
$\{10,20,31\}$&0.0098&20.8\\
$\{15,31\}$&0.0140&21.3\\
\hline
\end{tabular}
\end{center}
\label{table 3}
\end{table}

\subsection{Identifying the maximum safe dose in toxicological studies} \label{sec:closed.ex}
\citet{Tamhane01} describe a novel multiple testing approach to determining the maximum safe dose (MAXSD) of crop protection products such as pesticides, herbicides, and fungicides, which are tested for safety on non-target species, and for which the multiple testing error control guarantees a prescribed bound on recommending an unsafely high dose; their approach is equally applicable to clinical trials for safety with human subjects.  In this section we apply the rejection principle developed above to derive a sequential version of this procedure.

Assume there are $k$ discrete nonzero dose levels which we label $1,\ldots,k$ in increasing order, and include the level $0$ to denote ``no treatment.'' Following \citet[][Section~3]{Tamhane01}, we adopt an ANOVA setup in which there are $k+1$ groups of subjects, each treated at one of these dose levels. Let $y_{ij}$ ($i=1,\ldots,n_i$, $j=0,\ldots,k$) denote the response of the $i$th subject in the $j$th group, and let $\mu_j=E(y_{ij})$ ($j=0,\ldots,k$). 
Large values of $y_{ij}$ indicate safety of the treatment.  For example, in the crop protection setting mentioned above, $y_{ij}$ may be the growth of a non-target organism when exposed to the potentially toxic  treatment. Define the hypotheses
$$H^{(j)}: \mu_j\le \lambda \mu_0\quad\text{ vs. }\quad G^{(j)}: \mu_j>\lambda \mu_0, \quad j=1, \dots, k,$$
in which $\lambda\in(0,1]$ is a fixed, agreed-upon response threshold for safety, and therefore the null hypothesis $H^{(j)}$ means that the $j$th dose is unsafe; the MAXSD is  defined as  the largest $j\in\{0,\ldots,k\}$ such that $H^{(j)}$ is false. \citet{Tamhane01} propose a multiple testing approach to encode the natural ordering $\mu_0\ge \mu_1\ge\ldots\ge \mu_k$, as follows. Replacing $H^{(j)}$ by $$\wtilde{H}^{(j)}=\bigcap_{j'=j}^k H^{(j')},$$ we see that $\wtilde{H}^{(1)},\ldots, \wtilde{H}^{(k)}$ form a closed family of hypotheses and the test for closed hypotheses in Section~\ref{sec:closed} and Corollary~\ref{cor:closed} can be used to test these hypotheses sequentially. Since this procedure tests the hypotheses in order of decreasing dimension, as discussed in Section~\ref{sec:closed}, the hypotheses will be tested in the order $\wtilde{H}^{(k)}, \wtilde{H}^{(k-1)},\ldots, \wtilde{H}^{(1)}$. Here we focus on sequential control of only type~I FWER but, alternatively, simultaneous control of type~I and II FWERs could  be considered using the test of \citet{Bartroff14b} discussed in Section~\ref{sec:I&II}.  

The following simulation study was performed to explore the operating characteristics of the sequential testing procedure. Setting $k=4$ and taking $y_{ij}$ to be i.i.d.\ $N(\mu_i,1)$ observations, the mean responses~$\mu_i$ were chosen to be those in the second column of Table~\ref{table 4}, making the true MAXSD 1, indicted by an asterisk in the table. This choice of mean responses, with $\mu_1=0$, represents the commonly encountered but confounding testing situation in which the smallest nonzero dose actually has mean response zero, but is also the correct dose; we further note that recommending dose level~0 as the MAXSD, although it has the same mean response as dose level~1, has much different and possibly dangerous implications for the utilization of the MAXSD in future scientific work.  An $\alpha=.05$ version of both the sequential test in Section~\ref{sec:closed} and the fixed sample size version was implemented and Table~\ref{table 4}  gives the estimated average sample size (denoted Avg.\ SS) of each group and the probability (denoted by  $P(\mbox{MAXSD}=j)$) of choosing each dose level as the estimated MAXSD for both the sequential and fixed sample size procedures, based on $50,000$ Monte Carlo simulated data sets. For both of these tests, standard two-sample $t$ tests were used as the individual test statistics to test the $H^{(j)}$, with $\lambda$ taken to be 1, and the critical values were determined by Monte Carlo in the sequential case. The maximum sample size of the sequential procedure was 50, to mirror the sample size of the fixed sample size procedure.  At the three dose levels~$j=2, 3, 4$ exceeding the MAXSD, the sequential procedure only required on average 28.6, 8.9, and 2.6 observations, respectively, a dramatic reduction from the fixed sample size procedure which used 50 observations at each of these dose levels. While dosing subjects at levels above the MAXSD may not be a concern in some studies involving plants, etc., it would be of chief concern in clinical trials with human patients who are likely to experience toxicity, or even possibly death, at those levels. Both the sequential and fixed sample size procedures correctly identified the true MAXSD more than 80\% of the time, with the sequential procedure less likely to identify the true MAXSD than the fixed sample size procedure. On the other hand, the fixed sample size procedure was more likely to underestimate the MAXSD (4.70\%) compared to the sequential procedure (1.57\%), which is also undesirable but for different reasons such as an ineffective crop protection plan being implemented or sick human patients receiving ineffective treatment.  The last line of the table contains a weighted average estimated MAXSD for both tests.

\begin{table}[htdp]
\caption{Performance of the sequential and fixed sample size procedures for identifying the MAXSD.}\label{table 4}
\begin{center}
\begin{tabular}{|c|c|c|c|c|c|}

\hline
	&	&\multicolumn{2}{|c|}{Sequential}				&\multicolumn{2}{|c|}{Fixed-sample}\\\hline
Level $=j$&$\mu_j$	&Avg.\ SS	&$P(\mbox{MAXSD}=j)$	&Avg.\ SS&$P(\mbox{MAXSD}=j)$\\
\hline
%$0^*$&-&-&0.0002&-&0.0023\\
0\;&0&50.0&1.57\%&50&4.70\%\\
1$^*$&0&49.7&82.43\%&50&89.68\%\\
2\;&0.5&28.6&16.00\%&50&5.62\%\\
3\;&1.0&8.9&0\%&50&0\%\\
4\;&2.0&2.6&0\%&50&0\%\\
\hline
\multicolumn{2}{|c|}{Weighed average MAXSD}&\multicolumn{2}{|c|}{1.1}&\multicolumn{2}{|c|}{1.0}\\
\hline
\end{tabular}
\end{center}
\end{table}

\section{Discussion}\label{sec:disc}
We have given sufficient conditions, in the form of \eqref{monotonicity}-\eqref{singlestep}, for a sequential procedure to control the FWER, and have shown that they can be applied to testing situations where special structure is assumed, or not assumed. In addition to the rejection principle's utility in deriving new sequential procedures, it also provides a unified view of sequential FWER control. Although it remains an open question whether this rejection principle is also a necessary condition for FWER control, in any case the principle may  be useful for developing optimality theory of sequential multiple testing procedures, about which  little is known, even in the case of independent data streams. It seems likely that finding the most sequentially efficient procedure satisfying \eqref{monotonicity}-\eqref{singlestep} may be more attainable than finding the best sequential procedure controlling the FWER.

We have not gone into detail about applying the general sequential procedures discussed in Sections~\ref{sec:BL} and \ref{sec:I&II}, and we refer interested readers to the respective references \citep{Bartroff10e,Bartroff14b} for details. We will say here that both of these procedures can handle testing, in a given data stream~$j$, the commonly-encountered hypotheses of the form
\begin{equation}\label{t=t0}
H^{(j)}: \theta^{(j)}=\theta_0^{(j)}\qmq{vs.} G^{(j)}: \theta^{(j)}\ne\theta_0^{(j)},
\end{equation}
where $\theta^{(j)}$ is a possibly vector-valued parameter, and $\theta_0^{(j)}$ is some fixed value  of interest.  For example, in a parametric setup,  sequential generalized log-likelihood ratio statistics can be used and signed-root normal approximations \citep{Jennison97} or Monte Carlo can be used to compute critical values; see \citet{Bartroff14b} for details. Which of the procedures, in Sections~\ref{sec:BL} or  \ref{sec:I&II}, to use in practice may depend on the application. The procedure in Section~\ref{sec:BL} does not explicitly specify an acceptance rule (although, as mentioned there,  acceptances can be incorporated without violating the FWER control), and it therefore may be more appropriate in situations where ``early stopping'' is not as high a priority under the null~$H^{(j)}$ as under $G^{(j)}$, such as when $H^{(j)}$ represents a drug being safe or a process being ``in control.''  On the other hand, if early stopping is desirable under both $H^{(j)}$ and $G^{(j)}$, then the procedure of \citet{Bartroff14b} discussed in Section~\ref{sec:I&II} that controls both the type~I and II FWERs may be more appropriate.  In order to control the type~II FWER, this procedure naturally requires  control of the marginal type~II error rate in the form of \eqref{typeII}, which may not be possible with $G^{(j)}$ as written in \eqref{t=t0}, because values of $\theta^{(j)}\in G^{(j)}$ arbitrarily close to $\theta_0^{(j)}$ can make bounds like \eqref{typeII} impossible for any test. This can be remedied by constructing a surrogate alternative hypothesis $\wtilde{G}^{(j)}$ under which type~II error control is possible, for example $\wtilde{G}^{(j)}: ||\theta^{(j)}-\theta_0^{(j)}||\ge \delta$ for some $\delta>0$ and norm $||\cdot||$.  Here $\delta$ may represent the minimum significant separation of the parameter, or similar, and be well-motivated by the domain of application. On the other hand, the statistician could treat $\delta$ as a parameter to choose before testing in order to attain a sequential procedure with desirable operating characteristics, such as 
expected sample size. A similar analysis pertains to null hypotheses of the form $H^{(j)}: \theta^{(j)}\le \theta_0^{(j)}$ for scalar-valued parameter~$\theta^{(j)}$ and, more generally, of the form $H^{(j)}: u(\theta^{(j)})\le u_0^{(j)}$ for vector-valued $\theta^{(j)}\in\mathbb{R}^d$ and given  smooth function $u: \mathbb{R}^d\To \mathbb{R}$ and fixed scalar value~$u_0$; see \citet{Bartroff08c}, \citet{Bartroff08}, and \citet{Bartroff13} for examples of sequential generalized likelihood ratio tests for these situations.

In addition to the optimality theory mentioned above, another area of further work is to generalize the sequential sampling schemes. Above, we assumed that the set of possible streamwise sample sizes~$\mN$ is fixed in advance, but an alternative approach is to incorporate an efficient adaptive scheme, wherein the next sampling increment can be chosen as a function of the data, making the resulting procedures ``adaptive'' in yet another sense.  Adaptive sequential sampling schemes for hypothesis testing have been considered by many authors including \citet{Jennison06a} and \citet{Bartroff06b, Bartroff06, Bartroff07c}. Incorporating adaptive sampling schemes such as these into the multiple testing procedures is an exciting area of future research.

\section*{Acknowledgements}\label{sec:Ack} Bartroff's work was partially supported by  grant DMS-1310127 from the National Science Foundation and grant R01 GM068968 from the National Institutes of Health. 

%\bibliographystyle{apalike}
%\bibliography{../Bib_files/bibliography}

\def\cprime{$'$}

\end{document}